\documentclass[prodmode,permissions]{acmsmall-ec16-arxiv}

\pdfoutput=1

\usepackage[hidelinks]{hyperref}
\hypersetup{breaklinks=true}
\usepackage{lipsum, booktabs, cleveref}

\usepackage{tikz, wrapfig}

\renewcommand*{\le}{\leqslant}
\renewcommand*{\ge}{\geqslant}
\renewcommand*{\epsilon}{\varepsilon}
\renewcommand*{\emptyset}{\varnothing}

\newcommand{\pref}{\succcurlyeq}
\newcommand{\R}{\mathbb{R}}

\newcommand{\ER}{\ensuremath{\exists\R}}
\newcommand{\size}{\operatorname{size}}

\newcommand{\problem}[3]{\begin{center}\begin{tabular}{ll}
		\toprule 
		\multicolumn{2}{l}{#1} \\ \midrule 
		\textit{Instance:} & #2 \\ 
		\textit{Question:} & #3 \\
		 \bottomrule
	\end{tabular}\end{center}}

\usepackage[ruled]{algorithm2e}

\usepackage{amsmath,amsfonts,amssymb,bbm} 
\usepackage[numbers,sort&compress]{natbib} % for citet
\SetArgSty{textrm}  % for algorithm2e
\SetAlFnt{\small}
\SetAlCapFnt{\small}
\SetAlCapNameFnt{\small}
\SetAlCapHSkip{0pt}
\IncMargin{-\parindent}

\begin{document}

% Page heads
\markboth{Dominik Peters}{Recognising Multidimensional Euclidean Preferences}

% Title portion
\title{Recognising Multidimensional Euclidean Preferences} 
\author{DOMINIK PETERS
\affil{University of Oxford}
}
\begin{abstract}
Euclidean preferences are a widely studied preference model, in which decision makers and alternatives are embedded in $d$-dimensional Euclidean space. Decision makers prefer those alternatives closer to them. This model, also known as multidimensional unfolding, has applications in economics, psychometrics, marketing, and many other fields. We study the problem of deciding whether a given preference profile is $d$-Euclidean. For the one-dimensional case, polynomial-time algorithms are known. We show that, in contrast, for every other fixed dimension $d > 1$, the recognition problem is equivalent to the existential theory of the reals (ETR), and so in particular NP-hard. We further show that some Euclidean preference profiles require exponentially many bits in order to specify any Euclidean embedding, and prove that the domain of $d$-Euclidean preferences does not admit a finite forbidden minor characterisation for any $d > 1$. We also study dichotomous preferences
and the behaviour of other metrics, and survey a variety of related work.
\end{abstract}

%ACM is moving forward with the 2012 Classification system: http://dl.acm.org/ccs_flat.cfm. Please generate the CCSXML tex code through the online interactive system and insert the code below. 




\begin{bottomstuff}
Dominik Peters is supported by EPSRC.

Author' address: D. Peters, Department of Computer Science, University of Oxford; email: \url{dominik.peters@cs.ox.ac.uk}
\end{bottomstuff}

\maketitle

\section{Introduction}
\label{sec:intro}

The study of preferences spans a multitude of fields: economics (and particularly game theory and social choice), political science, psychology, multi-agent systems, marketing, and others. An important element of working with preferences is \emph{understanding} them by constructing models for them and identifying underlying structure. For example, in a psychological model of preferences, we aim to discover which underlying psychological process has generated the preferences we now observe \cite{coombs1964theory}. In political science, one might wonder about the underlying structure of `political space' by analysing voter preferences \cite{merrill1999unified}. In economics, working with well-structured preferences often allows a model to become analytically tractable.

This paper will mainly use the lens and language of \emph{computational social choice}, but our results also apply to a formal study of any of the concerns outlined above. In recent years, much work in computational social choice has focussed on identifying structure in a given preference profile. The reason for this is simple: many of the problems social choice aims to solve (such as preference aggregation, committee selection, or fair and efficient allocation)  are computationally hard. However, if we manage to identify underlying, hopefully low-dimensional structure in the input profile, we can exploit this structure to guide algorithms. This approach has been successful, particularly for the domain of \emph{single-peaked preferences}: once we have identified an \emph{axis} on which an input profile is single-peaked, we may efficiently find a Kemeny ranking \cite{brandt2010bypassing}, identify an optimal committee according to the Chamberlin-Courant rule \cite{betzler2013computation}, calculate possible winners for many voting rules \cite{faliszewski2009shield}, and solve the stable roommates problem \cite{bartholdi1986stable}.

This leaves the question of whether we can, in fact, efficiently find a certificate for single-peakedness (or another desired domain restriction) that such algorithms can use. Fortunately, the answer is often positive: there are efficient and certificate-producing recognition algorithms for preferences that are \emph{single-peaked} \cite{escoffier2008single}, \emph{single-crossing} \cite{elkind2012clone}, \emph{1-Euclidean} \cite{doignon1994polynomial}, or \emph{single-peaked on a tree} \cite{trick1989recognizing}. Indeed, each of these preference domains is quite well-understood, for example in terms of forbidden-substructure characterisations \cite{ballester2011characterization}, concise representations of all certificates \cite{peters2016preferences,bartholdi1986stable}, containment relations between the different domains \cite{elkind2014characterization}, and the probability with which a random preference profile falls within a given domain \cite{arxiv/BrunerL-likelihoodSP}.

Intriguingly, there is another preference domain for which we do not have a comparable amount of understanding, yet it is an extremely popular modelling choice across disciplines. Known as \emph{spatial preferences}, or as the \emph{$d$-Euclidean domain}, or as \emph{multidimensional unfolding}, this preference domain contains profiles that can be `embedded' into $d$-dimensional Euclidean space. Precisely, a preference profile is $d$-Euclidean if we can assign every voter and every alternative a point in $\R^d$ such that voters prefer those alternatives that are closer to them (according to the usual Euclidean metric) to those that are further away. This characterisation of preferences has intuitive appeal: considering $\R^d$ as a continuous `policy space', within which alternatives can vary along multiple dimensions, each voter is identified with an \emph{ideal point} \cite{bennett1960multidimensional}. The best alternative for the voter is the one that minimises the distance to the ideal policy. We could also think of a facility location problem, where a single facility needs to be placed somewhere in the plane, with each decision-maker preferring the facility to be placed as close to them as possible \cite{hotelling1929stability}.

In practice, the embedding of voters and alternatives into $\R^d$ is hidden, and we only have access to the ordinal ranking data in form of a preference profile \cite{anshelevich2015approximating}. Given this data, can we recover an appropriate embedding into $\R^d$ that explains the preferences? This problem is known as multidimensional unfolding \cite{bennett1960multidimensional}, and a large variety of methods that attempt to estimate embeddings have been proposed in the statistics and psychometrics literature [for a modern exposition see, e.g., \citeNP{borg2005modern}]. None of these methods is guaranteed to return a suitable embedding whenever it exists and terminate in polynomial time.

We analyse this problem from a formal computational perspective. 
In particular, we prove that the decision problem of identifying $d$-Euclidean preference profiles is NP-hard for each fixed $d\ge 2$. More precisely, we prove that the problem is equivalent to the existential theory of the reals (ETR), and thus is in fact $\ER$-complete. This means that the recognition problem is unlikely to be contained in NP (though it is decidable in PSPACE). Using recent results about hyperplane arrangements, we deduce that, for each fixed $d\ge 2$, there exist $d$-Euclidean preference profiles such that the coordinates of any $d$-Euclidean embedding require exponentially many bits to specify. Thus, there is provably no polynomial-time algorithm that, on input a preference profile, outputs a $d$-Euclidean embedding if one exists. 

While the domains of single-peaked and of single-crossing preferences admit characterisations by a finite list of forbidden configurations \cite{ballester2011characterization, bredereck2013characterization}, the 1-Euclidean domain does not admit such a characterisation \cite{ChenPW15}. \citet{ChenPW15} conjectured that the same is true for the domain of $d$-Euclidean preferences, for each fixed $d\ge 2$. We use a connection with the theory of oriented matroids to prove their conjecture.

The results in this paper cast doubt on the fruitfulness of exploiting structure in $d$-Euclidean preferences in computational social choice, and they also limit the extent to which the spatial preference model can be used to reliably explain observed preference data. Future work could explore ways of mitigating this situation. At the end of this paper, we briefly sketch one way that this could be done, namely by replacing the Euclidean $\ell_2$-metric by the $\ell_1$- or $\ell_\infty$-metric. We show that the recognition problems corresponding to these metrics are contained in NP.

\section{Related Work}
\label{sec:related-work}

Before we establish notations and hardness results, let us first survey work connected to the topics of this paper. 

In the psychometrics literature, the problem of detecting $d$-Euclidean preferences is known as \emph{multidimensional unfolding} \cite{hays1961multidimensional,bennett1960multidimensional}, a notion based on Coomb's \citeyear{coombs1950psychological,coombs1964theory} \emph{unidimensional unfolding} (which in our terminology corresponds to 1-Euclidean preferences). Many ways of solving the multidimensional unfolding problem have been proposed [e.g., \citeNP{roskam1968metric,kruskal1969geometrical,takane1977nonmetric}] which are based on iteratively fitting the data by optimising some badness-of-fit function. These methods face the problem of producing \emph{degenerate} embeddings in which voters are placed equidistantly to multiple alternatives. Some methods have attempted to partly remedy this problem [see, e.g., \citeNP{busing2005avoiding} and the references therein]. However, the papers in this area do not take a formal approach and typically do not establish formal guarantees concerning correctness and runtime bounds.

%\subsection{The one-dimensional case}
\label{subsec:one-dimension}

The domain of \emph{one}-dimensional Euclidean preferences exhibits a lot of structure and has many useful theoretical properties. In particular, every 1-Euclidean profile is both single-peaked and single-crossing \cite{grandmont1978intermediate}. In fact, in every 1-Euclidean embedding, the order of the alternatives forms a single-peaked axis, and the order of the voters forms a single-crossing order. Being a (strictly) stronger condition than both single-peakedness and single-crossingness [see, e.g., \citeNP{elkind2014characterization}], the 1-Euclidean preference domain inherits all of their good properties: 1-Euclidean profiles admit a transitive majority relation (and thus Condorcet winners), have the representative voter property, admit a non-manipulable voting rule, and a large number of important hard computational problems become tractable on this domain.

The combinatorial substructure of the 1-Euclidean domain has allowed researchers to devise polynomial-time algorithms that recognise this domain. The first such algorithm appears to have been found by \citet{doignon1994polynomial} in the context of unidimensional unfolding, and subsequently rediscovered by \citet{elkind2014recognizing}. This algorithm proceeds by finding a single-crossing ordering of the voters (which is unique up to reversal), then discovers a suitable ordering of the alternatives, and finally uses Linear Programming to check for an embedding into the real line making the input profile 1-Euclidean.
An alternative approach is described by \citet{knoblauch2010recognizing}. \citeauthor{knoblauch2010recognizing}'s algorithm exploits single-peakedness rather than the single-crossing condition: the algorithm first identifies a suitable single-peaked axis, and then again uses Linear Programming to produce the real embedding.

%\subsection{Euclidean dimension}

The \emph{Euclidean dimension} of a preference profile is the smallest integer $d$ such that the profile is $d$-Euclidean. Any preference profile over $m$ alternatives is $(m-1)$-Euclidean: to see this, place alternatives on the vertices of the $(m-1)$-simplex in $\R^{m-1}$; then any voter can be placed appropriately within the simplex to reflect the voter's preferences. \citet{bogomolnaia2007euclidean} study the question of which dimension is sufficient to represent all preference profiles of a certain size. They show that the dimension that suffices to embed all profiles with $n$ voters and $m$ alternatives is between $\min\{n-1,m-1\}$ and $\min\{n,m-1\}$. They also show that the Condorcet cycle on $k$ alternatives has Euclidean dimension $k-1$.

%\subsection{Counting}
\citet{kamiya2006arrangements, kamiya2011codimension} study the question of how many different voter-maximal $d$-Euclidean preference profiles there are for each fixed number $m$ of alternatives. They introduce and use `mid-hyperplane arrangements' to study this question, and answer it for the cases $d = 1$ and $d = m-2$. \citeauthor{kamiya2006arrangements} found the same problem for other $d$ to be ``quite difficult at this stage''. Our complexity result gives an indication that $d$-Euclidean preferences are much more chaotic when $d\ge 2$, giving an indication why their counting problem may be difficult to solve for $d\ge 2$.

%\subsection{Other preference domains, approximations}
Every 1-Euclidean profile is also single-peaked. Similarly, there seems to be a connection between profiles that are multidimensional single-peaked \cite{barbera1993generalized, sui2013multi} and ones that are multidimensional Euclidean, though we are not aware of a formal result to that effect. In recent years, there has been a lot of interest in profiles that are ``almost structured'' \cite{erd-lac-pfa:c:nearly-sp, cor-gal-spa:c:spwidth, bre-che-woe:c:nearly-restricted, elkind2014detecting}. To formalise the notion of closeness, a wide variety of metrics have been proposed (such as finding the minimum number of voters that need to be removed to obtain a structured profile, or finding a way of partitioning the profile into structured subprofiles). Let us note here that the Euclidean dimension of a profile could also be viewed as such a metric -- the smaller the dimension of the profile, the more structured it is. Of course, it follows from the results of this paper that it is hard to evaluate this metric. This hardness phenomenon has also been found for most other metrics [see the papers cited previously].

\section{Preliminaries}
\label{sec:preliminaries}

\subsection{Euclidean preferences}
\label{subsec:euclidean-prefs}

Let $A$ be a finite set of \emph{alternatives} or \emph{candidates}. A \emph{preference relation} $\pref$ over $A$ (usually referred to as a \emph{vote}) is a complete and transitive binary relation over $A$. We denote by $\succ$ the strict part of $\pref$, that is, $a \succ b$ if and only if $a \pref b$ but $b\not\pref a$; and we denote by $\sim$ the indifference part of $\pref$, with $a \sim b$ if and only if $a \pref b$ and $b \pref a$. A \emph{profile} $V$ over $A$ is a set of votes over $A$. For notational convenience, we give voters names like $v$ or $i$, and refer to their preference relation by $\pref_v$ and $\pref_i$.

Let $V$ be a profile over $A$. We say that $V$ is \emph{$d$-Euclidean} (where $d\ge 1$) if there exists a map $x : V \cup A \to \R^d$ satisfying
\begin{equation}
\label{eqn:euclidean}
a \succ_v b \implies \|x(v) - x(a)\| < \| x(v) - x(b)\| \qquad \text{for all $v\in V$ and $a,b\in A$.}
\end{equation}
Thus, voter $v$ prefers those alternatives which are closer to $v$ according to the embedding $x$. Here, $\|\cdot\|$ refers to the usual Euclidean $\ell_2$-norm on $\R^d$, that is
\[ \| (x_1,\dots,x_d) \| = \| (x_1,\dots,x_d) \|_2 = \sqrt{ x_1^2 + \cdots + x_d^2 }.  \]
The \emph{(open) ball} of radius $r$ centred at $c$ is $B(c,r) = \{ x \in \R^d : \| x - c \| < r \}$.

Typically, we will consider profiles of \emph{strict orders} in which every vote $\pref_v$ is antisymmetric (that is, we do not allow ties). Notice that in any $d$-Euclidean embedding of a profile of strict orders, no two alternatives can reside at the same point of $\R^d$.

In some applications, it makes sense to consider preferences that include indifferences (ties). An extreme case, which nevertheless finds many applications, is that of \emph{dichotomous preferences}. A vote $\pref$ is \emph{dichotomous} if there are no three alternatives $a,b,c\in A$ with $a \succ b \succ c$. Equivalently, $\pref$ is dichotomous if $A$ splits into approved and non-approved alternatives. That is, we can write $A = A_1 \cup A_2$ with $A_1 \cap A_2 = \emptyset$ satisfying $a \succ b$ iff $a \in A_1$ and $b\in A_2$. We then say that the voter \emph{approves} of the alternatives in $A_1$, while the voter does not approve the alternatives in $A_2$. A dichotomous vote can (and will) be specified by just giving the set $A_1$ of approved alternatives. 

Our definition of Euclidean preferences applies to dichotomous preferences as well. Following the terminology of \citet{elkind2015structured}, we call a profile of dichotomous preferences $d$-DE (Dichotomous Euclidean) if it is $d$-Euclidean. In this context, the definition requires that there is an embedding $x : V \cup A \to \R^d$ so that for each voter $v\in V$, the set of approved alternatives of $v$ coincides with the set of alternatives contained in some ball $B(x(v), r_v)$ centred at $x(v)$. We call a profile of dichotomous preferences $d$-DUE (Dichotomous \emph{Uniform} Euclidean) if there is an embedding $x : V \cup A \to \R^d$ so that for each voter $v\in V$, the set of approved alternatives of $v$ coincides with the alternatives contained in the \emph{unit} ball $B(x(v), 1)$ centred at $x(v)$.

Some authors define Euclidean preferences in a subtly different way from us. One popular definition involves reversing the direction of the implication arrow in (\ref{eqn:euclidean}). Notice that under this definition, whenever a voter is equidistant between two alternatives, the voter is free to break the tie in either way. In particular, when $d \ge 2$, then \emph{every} preference profile is ``$d$-Euclidean'' under this definition -- place all voters at the origin, and position alternatives on the unit sphere around the origin. In the area of multidimensional unfolding, embeddings of this type are said to include \emph{degeneracies} [see, e.g., \citeNP{busing2005avoiding}]. Our definition circumvents this issue by just outright disallowing degeneracies. 

\citet{bogomolnaia2007euclidean} use another different definition: they replace the implication in (\ref{eqn:euclidean}) by an if-and-only-if. This definition is equivalent to ours for strict preferences, but is much more restrictive for preferences including ties: \citeauthor{bogomolnaia2007euclidean}'s definition requires that whenever a voter $v$ is indifferent between $a$ and $b$, then $a$ and $b$ are equidistant to $v$. Our definition does not impose any relation on the relative distances in cases of ties.

\subsection{Forbidden Substructures}
Let $V$ be a profile over the alternative set $A$. If we delete some alternatives, and are left with the set $A' \subseteq A$, we can obtain the restricted profile $V|_{A'}$ where every vote is restricted in the obvious way: ${\succ}|_{A'} := {\succ} \cap (A' \times A')$. Now, if $V$ is a profile over $A$, and $W$ is a profile over $B$, we say that $V$ \emph{contains} $W$ if we can obtain $W$ by first deleting some alternatives and voters from $V$, and then relabelling the remaining alternatives and reordering the remaining voters. A preference domain $\mathcal R$ (that is, a set of profiles) may then be \emph{characterised by forbidden configurations} by giving a set $\mathcal S$ of profiles such that $\mathcal R = \{ \text{profile } V : V \text{ does not contain any $W\in\mathcal S$} \}$. 

We call a preference domain $\mathcal R$ \emph{hereditary} if it is closed under containment. That is, if $V \in \mathcal R$ and $V$ contains $W$, then $W\in\mathcal R$. The $d$-Euclidean domain is hereditary for any $d\ge 1$. Note that any hereditary domain $\mathcal R$ is characterised by its complement $\overline{\mathcal R}$ (that is, its set of counterexamples) in this way. However, a satisfying characterisation will either use a finite set of obstructions, or be otherwise highly structured.

\subsection{Existential theory of the reals (ETR)}
\label{subsec:etr}

The \emph{language} of the first-order \emph{theory of the reals} consists of formulas using as symbols (i) a countable collection of variable symbols $x_i$, (ii) constant symbols 0 and 1, (iii) addition, subtraction, multiplication symbols, (iv) the equality ($=$) and inequality ($<$) symbols, (v) Boolean connectives ($\lor,\land,\lnot$), (vi) universal and existential quantifiers ($\forall, \exists$). 
The \emph{theory} of the reals consists of all true sentences in this language, interpreted using the obvious semantics (where quantifiers quantify over the real numbers $\R$). 

The \emph{existential theory of the reals (ETR)} consists of the true sentences of the form
\[ \exists x_1\in \R \: \exists x_2 \in \R \dots \exists x_n \in \R \quad F(x_1,x_2,\dots,x_n) \]
with $F(x_1,x_2,\dots,x_n)$ a quantifier-free formula in the language just defined. In other words, $F$ is a Boolean combination of equalities and inequalities of real polynomials.

The \emph{decision problem} of ETR is the problem of deciding whether a given sentence of the above form is true, that is whether it is a member of ETR. \citet{schaefer2010complexity} introduced the complexity class $\ER$ as the class of decision problems that admit a polynomial-time many-one reduction to the decision problem of ETR. Thus, $\ER$ captures the computational complexity of the existential theory of the reals. We say that a problem $A$ is $\ER$-\emph{hard} if all problems in $\ER$ reduce to $A$ in polynomial time. We say that $A$ is $\ER$-\emph{complete} if it is contained in $\ER$ and is $\ER$-hard.

From the definition of ETR it is not even clear that the decision problem of ETR is decidable. By introducing a quantifier-elimination procedure, \citet{tarski1948decision} showed that ETR is in fact decidable. Since then, a variety of algorithmic improvement have been made over Tarski's procedure (which does not admit an elementary time bound), and there exist algorithms with a singly-exponential time dependence in the number of variables \cite{grigor1988complexity,renegar1992computational}. In addition, \citet{canny1988some} obtained the astonishing result that ETR can be solved in polynomial space. Thus $\ER \subseteq \text{PSPACE}$.

From the other direction, it is easy to see that ETR can be used to solve the propositional satisfiability problem (3SAT): we can encode that a variable $x$ is either true or false ($x = 0 \lor x = 1$), and we can encode a clause like $(x_1\lor \lnot x_2 \lor x_3)$ through $(x_1 + (1 - x_2) + x_3 \ge 1)$.\footnote{Some definitions of ETR and $\ER$ do not allow use of the equality symbol, but this makes no difference up to polynomial-time transformations \cite{SchaeferNash2015}.} Thus, ETR is NP-hard, and every $\ER$-hard problem is also NP-hard. Together, we have the containments $\text{NP} \subseteq \ER \subseteq \text{PSPACE}$.

Multiple \ER-complete problems are known, and many of them are questions of the form ``can a given combinatorial object be geometrically represented?''. Particular examples include recognising intersection graphs of line segments in the plane \cite{schaefer2010complexity}, of unit disk graphs \cite{kang2012sphere}, or of unit distance graphs \cite{schaefer2013realizability}. The problem of recognising $d$-Euclidean preference profiles falls exactly into this category: trying to find a geometric embedding that `explains' a given combinatorial structure. Another \ER-complete problem is a decision version of the problem of finding a Nash equilbrium of a non-cooperative game \cite{SchaeferNash2015}. In particular, it is \ER-complete to decide whether a given 3-player game has a Nash equilibrium within a given ball (in the simplex of mixed strategies). \ER-hardness also holds for various other decision problems related to Nash equilibria, and even restricted to symmetric games \cite{garg2015etr}.

\subsection{Arrangements of hyperplanes}
\label{subsec:arrangements-of-hyperplanes}
Our exposition and terminology follows \citet{kang2012sphere}.

An \emph{(affine) $d$-hyperplane} is a set of form $h = \{ x\in \R^d : c^Tx = b \} \subseteq \R^d$ for some $c \in \R^d$ and $b\in\R$. A particular example of a hyperplane, for given $p,q\in\R^d$, is the set $\{ x\in\R^d : \|x -p\| = \|x -q\| \}$ of points that are equidistant to $p$ and $q$; in two dimensions, this is the \emph{perpendicular bisector}. Any hyperplane $h$ divides $\R^d\setminus h$ into two connected components, namely the half-planes $c^Tx > b$ and $c^Tx < b$. We can give $h$ an \emph{orientation} by (arbitrarily) designating one of these components as $h$'s \emph{positive side} $h^+$, and the other as $h$'s \emph{negative side} $h^-$. We call a hyperplane with a chosen orientation an \emph{oriented hyperplane}. An \emph{oriented hyperplane arrangement} $(h_1,\dots,h_n)$ is a finite ordered collection of oriented hyperplanes in $\R^d$.

Given an oriented hyperplane arrangement $\mathcal H = (h_1,\dots,h_n)$, we can assign to each point $x\in\R^d$ its \emph{sign vector} $\sigma(x) \in \{-,0,+\}^n$ by setting
\[ \sigma(x)_i = \begin{cases}
+ & \text{if } x \in h_i^+ \\
0 & \text{if } x \in h_i \\
- & \text{if } x \in h_i^-.
\end{cases} \]
Thus, the sign vector of $x$ records, for each oriented hyperplane in the arrangement, on which side of the hyperplane $x$ lies. The \emph{combinatorial description} $\mathcal D(\mathcal H)$ of $\mathcal H$ is the collection of all sign vectors induced by the arrangement $\mathcal H$, that is \[\mathcal D(\mathcal H) = \{ \sigma(x) : x \in \R^d \}.\]
If $\mathcal D(\mathcal H) = \mathcal D(\mathcal H')$, then we say that $\mathcal H$ and $\mathcal H'$ are \emph{isomorphic}.

Every connected component of $\R^d \setminus \mathcal H = \R^d \setminus (h_1 \cup \cdots \cup h_n)$ is called a \emph{cell} (or \emph{chamber} or \emph{region}). All points in the same cell have the same sign vector.

\section{The Recognition Problem for Euclidean Preferences}
\label{sec:recognition}

In this section we will show that the problem of recognising $d$-Euclidean preferences is $\ER$-complete for each fixed $d\ge 2$. We will do this by reducing from a problem concerning arrangements of hyperplanes. But first let us formally define the relevant decision problem, and verify that the problem is in fact contained in $\ER$.

\problem
{$d$-EUCLIDEAN}
{set $A$ of alternatives, profile $V$ of strict orders over $A$}
{is $V$ $d$-Euclidean?}

\begin{proposition}
	\label{prop:ETR-containment}
	\textup{$d$-EUCLIDEAN} is contained in $\ER$ for every $d\ge 1$. In particular it is contained in \textup{PSPACE}.
\end{proposition}

\begin{proof}
	This is almost immediate from the definition of $d$-Euclidean preferences. Namely, a profile is $d$-Euclidean if and only if there \emph{exist} reals $x_{r,i}\in \R$ for each $r\in A\cup V$ and $i = 1,\dots,d$ such that whenever $a \pref_v b$, we have
	\begin{align*}
	& \| x_v - x_a\| < \| x_v - x_b\| \iff \sum_{i=1}^d \left(x_{v,i} - x_{a,i}\right)^2 < \sum_{i=1}^d \left(x_{v,i} - x_{b,i}\right)^2.
	\end{align*}
	Thus, the problem is equivalent to asking whether a system of polynomial inequalities has a solution. This system can be constructed in polynomial time, given the profile.
\end{proof}

This proposition in particular shows that $d$-EUCLIDEAN is \emph{decidable}, a fact that is not \emph{a priori} obvious, and to the best of our knowledge has not been previously noted.

Our starting point in the reductions is the following problem about combinatorial descriptions of hyperplane arrangements.

\problem
	{$d$-REALISABILITY}
	{a set $S \subseteq \{-,+\}^n$ of sign vectors with $(-,\dots,-),(+,\dots,+)\in S$}
	{is there an oriented $d$-hyperplane arrangement $\mathcal H$ with $S \subseteq \mathcal D(\mathcal H)$?}

\noindent	
For example, $S = \{$\scalebox{0.90}[1]{$(+,+,+,+), (-,+,+,-), (-,+,-,+), (-,+,-,-), (-,-,-,+), (-,-,-,-)$}$\}$ is 2-realised by the four red lines in \Cref{fig:euclidean}, where the red label of the line is placed on the positive side of the line.

\begin{theorem}[\textup{\citeNP{kang2012sphere}}]
	\hspace{-7pt}
	\textup{$d$-REALISABILITY} is $\ER$-complete for $d\!\ge\!2$.
\end{theorem}

\citeauthor{kang2012sphere} establish this by a reduction from SIMPLE STRETCHABILITY, the problem of deciding whether an arrangement of pseudolines can be stretched into an isomorphic arrangement of lines. That problem is $\ER$-complete by Mn{\"e}v's \citeyear{mnev1985realizability} universality theorem, a deep topological result about representing semialgebraic varieties. \citet{shor1991stretchability} gives a direct proof of NP-hardness by a reduction from SAT. \\

We are now ready to prove our main result.
	
\begin{theorem}
	\label{thm:total-etr-hard}
	\textup{$d$-EUCLIDEAN} is $\ER$-complete for each $d\ge 2$.
\end{theorem}
\begin{proof}
	We have already seen that \textup{$d$-EUCLIDEAN} is contained in $\ER$ (Proposition~\ref{prop:ETR-containment}). We now show $\ER$-hardness by a reduction from $d$-REALISABILITY. 
	
	\begin{figure}[t]
		\centering
		\scalebox{0.7}{
			\begin{tikzpicture}
			[line/.style={red!80!black, semithick},
			voter/.style={radius=0.05,black!90},
			region/.style={line width=1.13cm, black!6},
			alternative/.style={blue!75, radius=0.05, text=blue!75}]
			\clip (0,0) circle [radius=7cm];
			
			% alternative regions
			\draw[region] (0,0) circle [radius=3.4];
			\draw[region] (0,0) circle [radius=4.8];
			\draw[region] (0,0) circle [radius=6.2];
			\draw[text=black!80] (2.6,2.1) node {$R$};
			\draw[text=black!80] (3.75,3.0) node {$2R$};
			\draw[text=black!80] (4.85,3.9) node {$3R$};
			
			% coordinate system
			\draw[-stealth,thick] (-7, 0) -- (7, 0);
			\draw[-stealth,thick] (0, -7) -- (0, 7);
			
			% unit ball
			\draw (0,0) circle [radius=1];
			
			% lines
			\draw[line] (-4,-7) -- (4.2,7);
			\draw[line] (-7,3) -- (7,-3.5);
			\draw[line] (-7,-3) -- (7,2);
			\draw[line] (-2,7) -- (2.5,-7);
			
			% voters
			\fill[voter] (0.17,-0.2) circle;
			\fill[voter] (0.3,0.7) circle;
			\fill[voter] (-0.4,-0.3) circle;
			\fill[voter] (0.7, 0.2) circle;
			\fill[voter] (-0.62, 0.4) circle;
			\fill[voter] (0.23, -0.7) circle;
			
			% p-line, orthogonal corner
			\draw[line,blue,densely dotted] (-3.5,7) -- (2,-7);
			\draw[line,blue,densely dotted] (-7,-1.6) -- (7,2.4);
			
			% blue-red intersections
			\draw[blue] (-0.49, -0.67) circle [radius=0.07];
			\draw[blue] (0.12, 0.43) circle [radius=0.07];
			
			% alternatives
			% a1
			\fill[alternative] (-1.73,2.5) circle;
			\draw[alternative] (-1.5,2.8) node {$b_1$};
			%b1
			\fill[alternative] (0.697,-3.7) circle;
			\draw[alternative] (0.9,-3.45) node {$a_1$};
			% a2
			\fill[alternative] (4.25,1.62) circle;
			\draw[alternative] (4.4,1.95) node {$b_2$};
			%b2
			\fill[alternative] (-4.28,-0.83) circle;
			\draw[alternative] (-4.4,-1.15) node {$a_2$};
			% a3
			\fill[alternative] (-3.9,-5.0) circle;
			\draw[alternative] (-3.65,-5.3) node {$a_3$};
			%b3
			\fill[alternative] (2,5.5) circle;
			\draw[alternative] (1.75,5.75) node {$b_3$};
			
			% line labels
			\draw[line] (6.3,1.3) node {$h_1$};
			\draw[line] (-2.1,6.05) node {$h_2$};
			\draw[line] (5.55,-3.2) node {$h_3$};
			\draw[line] (-2.75,-5.7) node {$h_4$};
			\end{tikzpicture}}
		\caption{2-Euclidean embedding of a profile obtained from a 2-realisable sign vector set $S$ through the reduction of \Cref{thm:total-etr-hard}. Black dots represent voters $v_\sigma$, positioned within the unit ball and within the cell with sign vector $\sigma$ induced by the red hyperplanes (lines). The red labels are on the positive side of each line. Blue circles denote the points $p_i$. Blue dots correspond to alternatives; note that  $a_i$ and $b_i$ are at radius $Ri\pm 2$ from the origin.}
		\label{fig:euclidean}
	\end{figure}
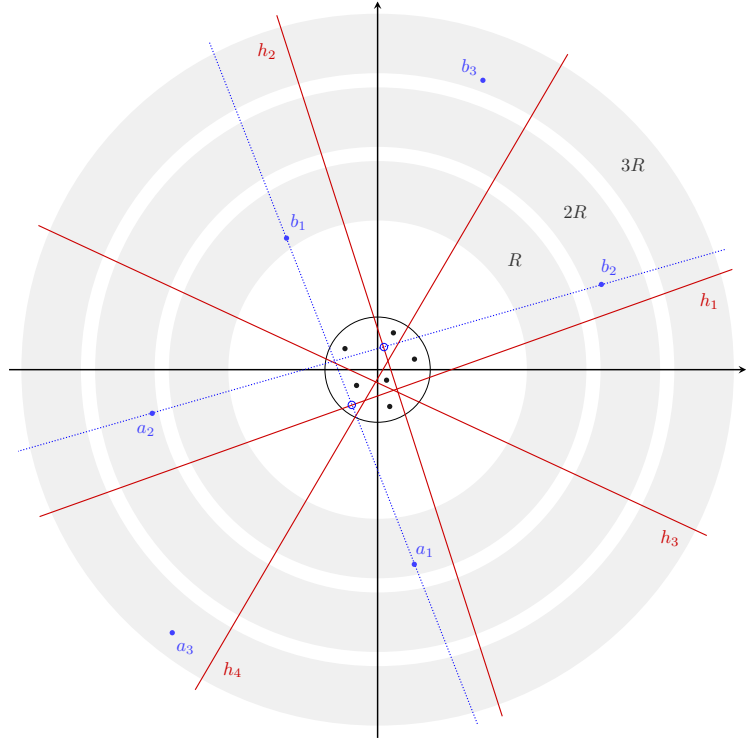
	
	Let $S \subseteq \{-,+\}^n$ be a given set of sign vectors with $(-,\dots,-),(+,\dots,+)\in S$. We construct a profile of $|S|$ votes over a total of $2n$ alternatives. Precisely, we take as alternatives the set $A = \{a_1,b_1,\dots,a_n,b_n\}$. For each $\sigma\in S$, we introduce a voter $v_\sigma$ with strict order $\pref_\sigma$ specified by
	\[ \{ a_1, b_1 \} \succ_\sigma \{a_2,b_2\} \succ_\sigma \cdots \succ_\sigma \{a_n,b_n\} \quad \text{and}\quad \begin{cases}
	a_i \succ_\sigma b_i \iff \sigma_i = +, \\
	b_i \succ_\sigma a_i \iff \sigma_i = -.
	\end{cases} \]
	This completes the description of the reduction. We now show its correctness. \\
	
	Suppose the profile constructed is $d$-Euclidean, and let $x : V\cup A \to \R^d$ be a Euclidean embedding. Take the oriented hyperplane arrangement $\mathcal H = (h_1,\dots,h_n)$ defined by
	\begin{align*}
	  h_i \: &= \{ x \in \R^d : \| x - x(a_i)\| = \| x - x(b_i)\| \}, \\
	  h_i^+ &= \{ x \in \R^d : \| x - x(a_i)\| < \| x - x(b_i)\| \}, \\
	  h_i^- &= \{ x \in \R^d : \| x - x(a_i)\| > \| x - x(b_i)\| \}.
	\end{align*}
	Then, clearly, $S\subseteq \mathcal D(\mathcal H)$: Let $\sigma\in S$ and let $i \in \{1,\dots,d\}$. If $\sigma_i = +$, we have $a_i \succ_\sigma b_i$, and thus by definition of Euclidean preferences, we must have $\| x(v_\sigma) - x(a_i) \| < \| x(v_\sigma) - x(b_i) \|$ and hence $x(v_\sigma) \in h_i^+$ so that $\sigma(x(v_\sigma))_i = + = \sigma_i$. Similarly if $\sigma_i = -$. It follows that $\sigma\in\mathcal D(\mathcal H)$. Hence $S\subseteq \mathcal D(\mathcal H)$. \\
	
	Conversely, suppose that $S\subseteq \mathcal D(\mathcal H)$ for some oriented $d$-hyperplane arrangement $\mathcal H$. By 
	%Lemma~\ref{lemma:cells-in-ball}, 
	applying an appropriate scaling map $x\mapsto \lambda x$ if needed,
	we may assume that every cell of $\mathcal H$ intersects the unit ball $B(0,1)\subseteq\R^d$. Write $\mathcal H = (h_1,\dots,h_n)$ with $h_i = \{ x\in\R^d : u_i^T x = b_i \}$, where without loss of generality $\|u_i\| = 1$, so that $u_i$ is a unit vector.  Further, we will say that $h_i^+ = \{ x : u_i^T x > b_i \}$ and $h_i^- = \{ x : u_i^T x < b_i \}$.
	
	We now construct a Euclidean embedding $x: A\cup V \to \R^d$. We start by placing the voter $v_\sigma$ corresponding to $\sigma\in S$ at an arbitrary point $x(v_\sigma) \in B(0,1)$ of the cell of $\mathcal H$ with sign vector $\sigma$. This exists by our assumption that $S\subseteq \mathcal D(\mathcal H)$. 
	
	Next, for each $i=1,\dots,d$, pick some point $p_i \in h_i \cap B(0,1)$ (this is possible because $B(0,1)$ meets both $h_i^-$ and $h_i^+$ since $(-,\dots,-),(+,\dots,+)\in S$). Following an argument by \citet{kang2012sphere}, we set for $r>0$
	\begin{align*}
	w_{i,r}^+ := p_i + r u_i \quad\text{and}\quad &w_{i,r}^- := p_i - r u_i, \\
	B_{i,r}^+ := B(w_{i,r}^+,r) \quad\text{and}\quad &B_{i,r}^- := B(w_{i,r}^-,r).
	\end{align*}
	\begin{minipage}{0.9\textwidth}
	Note that $B_{i,r}^+\subseteq h_i^+$ and $B_{i,r}^- \subseteq h_i^-$. In fact, $\bigcup_{r>0} B_{i,r}^+ = h_i^+$ and $\bigcup_{r>0} B_{i,r}^- = h_i^-$ (see figure on the right). Hence, for all $r$ sufficiently large, we have
	\begin{align*}
		&x(v_\sigma) \in B_{i,r}^+ \quad
		\text{for all $\sigma$ with $\sigma_i = +$, and} 
		%\text{for all $v_\sigma$ such that $a_i \succ_\sigma b_i$, and} 
		\\
		&x(v_\sigma) \in B_{i,r}^- \quad
		\text{for all $\sigma$ with $\sigma_i = -$.} 
		%\text{for all $v_\sigma$ such that $b_i \succ_\sigma a_i$.}
	\end{align*}
	Fix a value $R>4$ of $r$ for which this holds. We now pick the \\ positions of the alternatives in the Euclidean embedding: Set
	\end{minipage}
	\begin{minipage}{0.09\textwidth}
	%\begin{wrapfigure}[10]{l}{0pt}
		%\vspace{-20pt}
		\hspace{-58pt}
		\scalebox{0.6}{
			\begin{tikzpicture}
			[ball/.style={blue!80!black}]
			\clip (-2.5,-2.5) rectangle (2.5,2.5);
			
			%circles
			\draw[ball] (0.5,-0.5) circle [radius=0.7071cm];
			\draw[ball] (1,-1) circle [radius=1.414cm];
			\draw[ball] (2,-2) circle [radius=2.828cm];
			\draw[ball] (5,-5) circle [radius=7.07cm];
			
			% line
			\draw [very thick, red!80!black] (-3,-3) -- (3,3);
			
			%labels
			\draw (-0.2,0.2) node {$p_i$};
			\fill (0, 0) circle [radius=0.3mm];
			\draw (0.85,-0.65) node {\small $w_{i,1}^-$};
			\fill (0.5,-0.5) circle [radius=0.3mm];
			%\draw (1.7,-1.5) node {$w_{i,2}^-$};
			\end{tikzpicture}
		}
		%\vspace{5pt}
	%\end{wrapfigure}
	\end{minipage}
	\[ x(a_i) = w_{i, Ri}^+ \quad\text{and}\quad x(b_i) = w_{i, Ri}^-. \]
	We are left to verify that the map $x: A\cup V \to \R^d$ thus constructed actually corresponds to voters' preferences. First let us show that, according to the embedding $x$, every voter's preference has the form
	\[\{ a_1, b_1 \} \succ_\sigma \{a_2,b_2\} \succ_\sigma \cdots \succ_\sigma \{a_n,b_n\}. \]
	So let $v_\sigma$ be a voter, let $1 \le i < j \le n$, and let $c_i \in \{a_i,b_i\}$ and $c_j \in \{a_j, b_j\}$. Then
	\begin{align*}
		\| x(v_\sigma) - x(c_i) \| &\le \| x(v_\sigma) - p_i \| + \| p_i - x(c_i)\| \tag{triangle inequality} \\
		&\le 2 + Ri \tag{$u_i$ is a unit vector} \\
		&< Rj - 2 \tag{$j > i$ and $R > 4$}\\
		&\le \| x(c_j) - p_j \| - \| p_j - x(v_\sigma) \|  \tag{as before} \\
		&\le \| x(v_\sigma) - x(c_j) \|. \tag{reverse triangle inequality}
	\end{align*}
	Thus, it follows that $c_i \succ_\sigma c_j$, as desired.
	Finally, we need to confirm that
	\[\| x(v_\sigma) - x(a_i) \| < \| x(v_\sigma) - x(b_i) \| \iff \sigma_i = +.\]
	So suppose $\sigma_i = +$. By choice of $R$, we have $x(v_\sigma) \in B_{i,Ri}^+$, so that $\| x(v_\sigma) - x(a_i) \| < Ri$. On the other hand, we have $\| x(v_\sigma) - x(b_i) \| \ge Ri$: for suppose not. Then $x(v_\sigma) \in B_{i,Ri}^-$, and thus $x(v_\sigma) \in B_{i,Ri}^+ \cap B_{i,Ri}^- = \emptyset$, a contradiction.
	%\{p_i\}$. But $p_i$ lies \emph{on} the hyperplane $h_i$, and so $p_i$ is not part of a cell. Thus $x(v_\sigma)$ is not part of a cell, contradicting our choice of $x(v_\sigma)$. The case when $\sigma_i = -$ is similar, and this gives the result.
\end{proof}

Certainly, this hardness result implies that it is also hard to recognise $d$-Euclidean profiles of \emph{weak orders} (since strict orders form a special case). For \emph{dichotomous} orders, hardness does not follow immediately, but a similar reduction can be used. The decision problems for dichotomous preferences are defined as follows:

\problem
{$d$-DICHOTOMOUS-EUCLIDEAN}
{set $A$ of alternatives, profile $V$ of dichotomous votes over $A$}
{is $V$ $d$-DE?}

\problem
{$d$-DICHOTOMOUS-UNIFORM-EUCLIDEAN}
{set $A$ of alternatives, profile $V$ of dichotomous votes over $A$}
{is $V$ $d$-DUE?}

Perhaps unsurprisingly, the argument employed is almost identical to the hardness result for recognising unit disk graphs \cite{kang2012sphere}. 

\begin{theorem}
	Both \textup{$d$-DICHOTOMOUS-EUCLIDEAN} and \textup{$d$-DICHOTOMOUS-UNIFORM-EUCLIDEAN}  are $\ER$-complete for each $d\ge 2$.
\end{theorem}

\begin{proof}[Sketch]
	The proof is similar to the previous reduction. We again reduce from $d$-REALISABILITY. The same reduction works for both DE and DUE. Let $S \subseteq \{-,+\}^n$ be a given set of sign vectors with $(-,\dots,-),(+,\dots,+)\in S$. We construct a profile of $|S|$ dichotomous votes over the same set of alternatives $A = \{a_1,b_1,\dots,a_n,b_n\}$ as before. For each $\sigma\in S$, the voter $v_\sigma$ has dichotomous preferences approving the set
	\[ v_\sigma = \{ a_i : i = 1,\dots,n \text{ with } \sigma_i = + \} \cup \{ b_i : i = 1,\dots,n \text{ with } \sigma_i = - \}.\]
	This completes the description of the reduction.
	
	Correctness can be established using almost the same argument as in \cite[Theorem~1]{kang2012sphere}; we omit the details here for lack of space. Compared to the argument for \Cref{thm:total-etr-hard}, we need to do slightly more scaling and shifting.
\end{proof}

\section{Precision}
\label{sec:precision}

In this section, we consider the question of how many bits are needed to specify a Euclidean embedding $x : A \cup V \to \R^d$. We only consider the `natural encoding' where the coordinates of each point are given as rational numbers. Note that if every $d$-Euclidean profile were to admit an embedding that can be specified in polynomially many bits, then this would put the problem $d$-EUCLIDEAN in NP. Yet in this section we show that there is a family of profiles which need exponentially many bits in order to specify any Euclidean embedding. This result, by itself, does not rule out that the \emph{decision} problem $d$-EUCLIDEAN is in NP: there could be a `clever' way to certify that an embedding exists, without explicitly giving the embedding (finding such a `clever' certificate would prove $\text{NP} = \ER$, which in the words of \citet{kang2012sphere} would constitute a ``minor breakthrough in complexity theory''). On the other hand, our result shows that the \emph{function} problem associated with the problem $d$-EUCLIDEAN is \emph{provably} not in P.

Let us now make precise the notion of the \emph{size} of an embedding $x : A \cup V \to \R^d$. Here, we follow the definitions of \citet{mcdiarmid2013integer}. The number of bits needed to store a natural number $n\in\mathbb N$ is the number of digits in its binary representation: $\size(n) := \lceil \log_2(n+1) \rceil$. To represent an integer $k\in\mathbb Z$, we need an extra bit to store its sign: $\size(k) := 1 + \size(|k|)$. Finally, we represent a rational number $q\in\mathbb Q$ as a pair of integers representing a fraction: if $q = m/n$, where $m,n\in\mathbb Z$ are relatively prime, we set $\size(q) := \size(m) + \size(n)$. The size of a rational vector $x \in \mathbb Q^d$ is $\size(x) := \sum_{i=1}^d \size(x_i)$. Then, the size of a \emph{rational} Euclidean embedding $x : A \cup V \to \mathbb Q^d$ is defined as
\[ \size(x) := \sum_{r \in A \cup V} \size(x(r)). \]

Before we establish the promised lower bound, let us first give a corresponding upper bound. Namely, while some $d$-Euclidean profiles require exponentially many bits to specify, (single-)exponentially many bits are always enough. To see this, we will first need a guarantee that every $d$-Euclidean profile admits a rational embedding, because we have only assigned sizes to rational embeddings.

\begin{theorem}
	\label{thm:upper-bound}
	Every $d$-Euclidean profile admits a rational embedding. Further, for each $d \ge 1$, there is a constant $c = c(d)$ such that any $d$-Euclidean profile with $n$ voters and $m$ alternatives admits a rational embedding $x : A \cup V \to \mathbb Q^d$ with $\size(x) \le 2^{c(n+m)}$. 
\end{theorem}

This theorem is an essentially immediate corollary of the following general result about the bit sizes of solutions to polynomial inequalities.

\begin{theorem}[\textup{\citeNP{basu1996combinatorial}}]
	\label{thm:poly-systems}
	Fix $d,\tau\in\mathbb N$. There is a constant $C = C(d, \tau)$ such that for all sets $\mathcal P$ of polynomials in $n$ variables of degree at most $d$ and with integer coefficients of bit size at most $\tau$, we have that whenever the system $\{p(x) > 0 : p\in\mathcal P \}$ has a real solution, then it has a rational solution of bit size at most $|\mathcal P| \cdot \tau \cdot d^{Cn}$.
\end{theorem}

\begin{proof}[of \Cref{thm:upper-bound}]
	From Proposition~\ref{prop:ETR-containment}, we know that $d$-Euclidean embeddings are precisely the solutions to a certain system of strict polynomial inequalities, in which all polynomials have degree 2, and all coefficients are at most 2 in absolute value. So we can apply \Cref{thm:poly-systems} to obtain our result.
\end{proof}

The upper bound of \Cref{thm:upper-bound} is not tight for the case $d=1$. Recall from \Cref{subsec:one-dimension} that there are polynomial-time algorithms for recognising the 1-Euclidean domain. These work through a (non-trivial) reduction to linear programming. The linear programs produced in this reduction have polynomially bounded coefficients (in fact, bounded by 2), and are thus \emph{combinatorial} linear programs, and so admit a \emph{strongly polynomial} algorithm \cite{tardos1986strongly}. Thus, we can say the following:

\begin{proposition}
	Every 1-Euclidean profile with $n$ voters and $m$ alternatives admits a rational embedding $x : A \cup V \to \mathbb Q$ with $\size(x)$ bounded by $\operatorname{poly}(n,m)$. 
\end{proposition}

For the lower bound, we use techniques developed by \citet{mcdiarmid2013integer} and \citet{kang2012sphere} and apply them to the reduction of \Cref{thm:total-etr-hard}. For a profile $V$ over alternative set $A$, we define $\|V\| := |V|+|A|$.

\begin{theorem}
	Fix $d\ge 2$. For a $d$-Euclidean preference profile $V$, let $e(V)$ denote the minimum size of a rational Euclidean embedding of $V$. For each $m\ge 1$, let $e(n+m)$ be the maximum $e(V)$ among $d$-Euclidean preference profiles $V$ with $\|V\| = n+m$. Then $e(n+m) \ge 2^{\Omega(n+m)}$.
\end{theorem}

\begin{proof}[Sketch]
	\citet{mcdiarmid2013integer} construct a family of combinatorial descriptions of line arrangements that have doubly-exponential `\emph{span}', which (roughly) corresponds to the size of the numbers needed to represent any realisation of the combinatorial description. \citet{kang2012sphere} generalise this construction to $d$-dimensional hyperplane arrangements. They then transform these line arrangements into unit disk graphs in the proof of their Theorem 3. That proof goes through almost verbatim for our case (using the reduction of our \Cref{thm:total-etr-hard}). At one step in the proof, \citeauthor{kang2012sphere} need to introduce isolated vertices to pad the graph in question. In our setting, we may introduce new alternatives and add them to the bottom of each voter's preference list; in a Euclidean embedding we just place these alternatives far away from all the voters.
\end{proof}

\section{Forbidden Minor Characterisations}
\label{sec:minors}

Consider a possible characterisation of the $d$-Euclidean domain by a set $\mathcal S$ of forbidden configurations. We will call this characterisation \emph{good} if the set $\mathcal S$ is polynomial-time recognisable: that is, given a profile, there should be a polynomial-time algorithm deciding whether the given profile is one of the configurations contained in $\mathcal S$. Certainly, if $\mathcal S$ were \emph{finite}, then $\mathcal S$ provides a good characterisation. However, there exist infinite characterisations that are still good in this sense, for example for interval graphs \cite{lekkeikerker1962representation} and matrices with the consecutive ones property \cite{tucker1972structure}.
(It could be argued that in order for $\mathcal S$ to be good it needs to be recognisable in LOGSPACE or another complexity class below P.)

However, given the complexity result of \Cref{sec:recognition}, it is a straightforward observation that for each $d \ge 2$, no good characterisation by forbidden substructures will exist for the $d$-Euclidean domain, subject to a reasonable complexity-theoretic assumption.

\begin{proposition}
	For each $d\ge 2$, the set of $d$-Euclidean preference profiles does not admit a good characterisation by forbidden substructures unless $\ER \subseteq \textup{coNP}$.
\end{proposition}

Note that $\ER \subseteq \textup{coNP}$ would imply $\textup{NP} \subseteq \textup{coNP}$, itself a rather unlikely event.

\begin{proof}
	Suppose a good characterisation by $\mathcal S$ exists. We give a coNP-algorithm that recognises $d$-Euclidean preferences: Given an input profile, guess some subprofile, guess a relabeling of voters and alternatives in this subprofile, and check whether the result is contained in $\mathcal S$.
\end{proof}

By a similar argument, no \emph{finite} characterisation can exist unless $\textup{P} = \ER = \textup{NP}$. In the remainder of this section, we prove this weaker result \emph{without} appealing to any complexity-theoretic assumptions. To do this, we use a connection between the theory of arrangements of hyperplanes and the theory of \emph{ordered matroids}.

\begin{theorem}[\textup{\citeNP{bokowski1989infinite}}]
	\label{thm:matroids}
	There exist infinitely many nonrealisable uniform oriented matroids of rank 3 such that every proper minor of them is realisable. In particular, by the Topological Realisation Theorem, for every $n_0\in \mathbb N$, there exists a non-stretchable simple pseudoline arrangement with $n > n_0$ lines such that removing any line results in a stretchable arrangement.
\end{theorem}

To prove our result about $d$-Euclidean preferences, we will use the examples from \Cref{thm:matroids} and apply to them the chain of many-one reductions that yielded the hardness result of \Cref{thm:total-etr-hard}. We will present this argument in several lemmas.

\begin{lemma}
	\label{lemma:minors-realizability}
	Fix $d \ge 2$. For every $n_0\in\mathbb N$, there exists $n>n_0$ and a set $S\subseteq \{-,+\}^n$ with $(+,\dots,+),(-,\dots,-)\in S$ such that $S$ is \emph{not} $d$-realisable, but for each $i = 1,\dots,n$, the set $S_{-i} = \{ s_{-i} : s \in S \}$ obtained by deleting coordinate $i$ \emph{is} $d$-realisable.
\end{lemma}
\begin{proof}
	For $d = 2$, this is a direct consequence of \Cref{thm:matroids} after applying the relabelling procedure described in the proof of Theorem~10 of \cite{kang2012sphere} -- relabelling is necessary to ensure that $(+,\dots,+),(-,\dots,-)\in S$. 
	
	We are left to show the result for $d > 2$, and we proceed by induction. For this, we will need Lemma~11 of \cite{kang2012sphere} which states that a set $S\subseteq \{-,+\}^n$ with $(+,\dots,+),(-,\dots,-)\in S$ is $d$-realisable if and only if $S \times \{+,-\}$ is $(d+1)$-realisable. Now, let $n_0\in\mathbb N$ be given, and use the inductive hypothesis to find $n \ge n_0+2$ and $S\subseteq \{-,+\}^n$ such that $S$ is not $d$-realisable, but deleting any coordinate yields a $d$-realisable set. Now consider $S' := S \times\{+,-\}$. By the result quoted, $S'$ is not $(d+1)$-realisable, but for $i = 1,\dots,n$, the sets $S'_{-i}$ are $(d+1)$-realisable. If $S'_{-(n+1)}$ ($=S$) also happens to be $(d+1)$-realisable, then we are done, since in this case $S'$ is minimally non-$(d+1)$-realisable. So suppose that $S'_{-(n+1)}$ is \emph{not} $(d+1)$-realisable. In this case, $S'_{-(n+1)}$ is minimally non-$(d+1)$-realisable, since deleting any further coordinate leaves a minor of one of the $S'_{-i}$ which we know to be $(d+1)$-realisable. In either case, we have found a minimal counterexample of size greater than $n_0$, as required.
\end{proof}

\begin{lemma}
	\label{lemma:euclid-minors}
	Fix $d \ge 2$. For every $m_0\in\mathbb N$, there is $m>m_0$ such that there exists a preference profile over $m$ alternatives which is not $d$-Euclidean, yet removing any alternative yields a $d$-Euclidean profile.
\end{lemma}

\begin{proof}
	Let $m_0\in\mathbb N$ be given, and find $m'>m_0$ and $S \subseteq \{+,-\}^{m'}$ satisfying the conditions of Lemma~\ref{lemma:minors-realizability}. From $S$, construct the preference profile $V$ over alternative set $\{a_1,b_1\dots,a_{m'}, b_{m'}\}$ with  $2m' > m_0$ alternatives as in the proof of \Cref{thm:total-etr-hard}.
	
	According to the proof of \Cref{thm:total-etr-hard}, $V$ cannot be $d$-Euclidean since $S$ is not $d$-realisable. The profile $V'$ obtained from the $d$-realisable set $S_{-i}$, however, \emph{is} $d$-Euclidean: From the definition of the reduction it is clear that $V'$ is just $V$ with the alternatives $a_i$ and $b_i$ removed. Suppose only one of these is removed: we need to argue that the resulting profile is also $d$-Euclidean. Without loss of generality, we remove $a_i$ from the profile $V$ to obtain profile $V''$. Take the $d$-Euclidean embedding of $V'$ that is produced in the proof of \Cref{thm:total-etr-hard}, and place alternative $b_i$ at any point $x(b_i)$ that is distance $Ri$ away from the origin. Just like in the proof of \Cref{thm:total-etr-hard}, we can see that this embedding makes $V''$ $d$-Euclidean.
\end{proof}

It is worth noting that all these infinitely many minimal counterexamples have the shape $\{ a_1, b_1 \} \succ \{a_2,b_2\} \succ \cdots \succ \{a_n,b_n\}$, and in particular they are single-peaked. 

\begin{theorem}
	The domain of $d$-Euclidean preferences does not admit a finite characterisation by forbidden configurations, for any fixed $d\ge 1$.
\end{theorem}
\begin{proof}
	For $d = 1$, this is the main result of \citet{ChenPW15}. For fixed $d\ge 2$, suppose for a contradiction that such a characterisation exists, and let $M$ be the maximum number of alternatives in any of the forbidden configurations. By Lemma~\ref{lemma:euclid-minors}, there exists a profile $V$ over at least $M+1$ alternatives which is not $d$-Euclidean. Since the forbidden configurations characterise the $d$-Euclidean domain, one of the configurations must be contained in $V$. In fact, considering the size of the configuration, it must be contained in $V$ even with one alternative deleted. However this profile \emph{is} $d$-Euclidean, contradicting the fact that the $d$-Euclidean domain is hereditary.
\end{proof}

\section{Other metrics}
\label{sec:other-metrics}

In \Cref{subsec:euclidean-prefs}, we defined Euclidean preferences using the usual Euclidean $\ell_2$-metric, measuring distances by shortest paths in the plane. Other choices of metric may be preferred in certain contexts, and in this section we will briefly consider the effect of using other metrics on the complexity of the recognition problem.

The two metrics we consider here are the $\ell_1$-metric and the $\ell_\infty$-metric. The $\ell_1$-metric is also known as the \emph{cityblock} or \emph{taxicab} or \emph{Manhattan} distance, because it measures distances by shortest paths on a grid like the street network of Manhattan. Formally, the $\ell_1$-norm is defined by
\[ \|(x_1,\dots,x_d)\|_1 := |x_1| + \cdots + |x_d|. \]
Thus, the $\ell_1$-distance $\| x - y \|_1$ of two points $x$ and $y$ is the sum of the absolute distances along each coordinate axis. The $\ell_\infty$-metric, on the other hand, measures the \emph{maximum} distance along a coordinate axis:
\[ \|(x_1,\dots,x_d)\|_\infty := \max \{|x_1|, \dots, |x_d|\}. \]
For each of these metrics (or indeed, any metric space), we can obtain a notion of $d$-Euclidean preferences by just plugging this metric into the definition in line (\ref{eqn:euclidean}). For certain settings, the $\ell_\infty$ metric has a nice interpretation as corresponding to `pessimistic' voters who judge candidate according to their (subjectively) worst feature. The $\ell_1$ metric also has intuitive appeal -- see \citet{eguia2011foundations} and the references therein for arguments in favour of using Euclidean preferences with respect to this metric.

Comparing the $\ell_1$- and $\ell_\infty$-metrics to the $\ell_2$-metric we have used so far, one gets the sense that $\ell_1$ and $\ell_\infty$ are more `discrete' or `combinatorial' than the more geometric~$\ell_2$. Supporting this intuition, we find that the complexity of the recognition problem changes (unless $\textup{NP} = \ER$) when we use one of these metrics.

\begin{theorem}
	\label{thm:np-containment}
	The problems of recognising preference profiles that are $d$-Euclidean with respect to the $\ell_1$-metric or the $\ell_\infty$-metric are contained in \textup{NP} for every $d\ge 1$.
\end{theorem}
\begin{proof}
	We start with the $\ell_1$-metric and show containment in NP by giving a nondeterministic reduction to linear programming. For each of the $d$ coordinate axes of $\R^d$, nondeterministically guess in which order the points corresponding to voters and alternatives appear along that axis. 
	%TODO: wlog these orders can be taken to be strict
	Once we have decided these orderings, we can rewrite the definition of $\ell_1$-Euclidean preferences without using absolute values. Then, we can replace strict inequalities with weak inequalities by introducing additive `slack' constants  \cite[Prop.~3]{elkind2014recognizing}. The result is a linear (feasibility) program, which produces a suitable $\ell_1$-Euclidean embedding if one exists. 
	\begin{alignat}{3}
	d_{v,c,i} &= \pm (x_{v,i} - x_{c,i}) \quad && \text{(distance between $v$ and $c$ along $i$)} \label{eqn:plus-minus} \\
	d_{v,c} &= \textstyle\sum_{i=1}^d d_{v,c,i} && \text{(distance between $v$ and $c$)} \\
	d_{v,a} &\le d_{v,b} - 1 && \text{when $a \succ_v b$} \\
	x_{y,i} &\le x_{y',i} - 1 && \text{when $y$ occurs to the left of $y'$ on axis $i$} 
	\end{alignat}
	In constraints of form (\ref{eqn:plus-minus}), the $\pm$ can be replaced by plus or minus at `compile'-time so that the quantity reflects $|x_{v,i} - x_{c,i}|$.
	
	The argument for the $\ell_\infty$-metric is similar: here we additionally guess for each pair $(v,c)$ in which direction the maximum distance is achieved.
\end{proof}

\section{Conclusions}
The results of this paper are bad news for the $d$-Euclidean domain: because producing a Euclidean embedding will in general be infeasible, we are stuck with heuristic algorithms that may or may not produce a correct output in their allotted time. In some sense, our hardness results show that the estimation algorithms developed for the multidimensional unfolding problem over the past several decades are best possible, in the sense that we cannot hope for exact efficient algorithms. Still, future developments in ETR-solver technology might allow solving practical instances in reasonable time, and perhaps some of the ideas in the area of multidimensional unfolding can be formalised and yield exact algorithms. We have run some preliminary experiments on the PrefLib dataset \cite{mattei2013preflib} using the \texttt{nlsat} solver \cite{jovanovic2012solving} which is part of the \texttt{z3} theorem prover \cite{de2008z3}, and appears to be the strongest ETR-solver available. However, \texttt{nlsat} was unable to decide whether \emph{any} of the PrefLib profiles that we tried was 2- or 3-Euclidean within a time bound of one hour (except for trivial profiles on 3 alternatives).

The general infeasibility of identifying membership in the $d$-Euclidean domain ($d\ge 2$) also means that any efficient algorithm developed for actual voting problems that exploits the spatial structure will need to be given a Euclidean embedding as part of the input. Except perhaps for facility-location type problems, not too many examples come to mind in which the underlying spatial structure is known \emph{a priori}. Thus, such algorithms may turn out to be of limited use.

While the multidimensional case seems nasty, there is hope that we will be able to develop a better understanding of the one-dimensional Euclidean domain in the future. We would like to reiterate here two open problems posed elsewhere in the literature: \citet{elkind2014recognizing} ask whether the 1-Euclidean domain can be recognised by a `combinatorial' algorithm that does not rely on solving a linear program, and \citet{ChenPW15} ask whether there is an explicit and good characterisation of the 1-Euclidean domain by (infinitely many) forbidden configurations. It appears likely that the answer to both of these questions is the same, and a positive answer would require a better structural understanding of the 1-Euclidean domain.

Finally, the versions of the problem we discussed in \Cref{sec:other-metrics} for $\ell_1$ and $\ell_\infty$ are intriguing: What is the precise complexity of the recognition problem? How are these notions related to multidimensional single-peakedness?

\begin{acks}
I thank Edith Elkind and Martin Lackner for helpful discussions, and J\"{u}rgen Bokowski for correspondence about the theory of oriented matroids.
\end{acks}

% Bibliography
\bibliographystyle{ACM-Reference-Format-Journals}

\begin{thebibliography}{00}
	
	%%% the Web address.
	
	\ifx \showCODEN    \undefined \def \showCODEN     #1{\unskip}     \fi
	\ifx \showDOI      \undefined \def \showDOI       #1{{\tt DOI:}\penalty0{#1}\ }
	\fi
	\ifx \showISBNx    \undefined \def \showISBNx     #1{\unskip}     \fi
	\ifx \showISBNxiii \undefined \def \showISBNxiii  #1{\unskip}     \fi
	\ifx \showISSN     \undefined \def \showISSN      #1{\unskip}     \fi
	\ifx \showLCCN     \undefined \def \showLCCN      #1{\unskip}     \fi
	\ifx \shownote     \undefined \def \shownote      #1{#1}          \fi
	\ifx \showarticletitle \undefined \def \showarticletitle #1{#1}   \fi
	\ifx \showURL      \undefined \def \showURL       #1{#1}          \fi
	
	\bibitem[\protect\citeauthoryear{Anshelevich, Bhardwaj, and Postl}{Anshelevich
		et~al\mbox{.}}{2015}]%
	{anshelevich2015approximating}
	{Elliot Anshelevich}, {Onkar Bhardwaj}, {and} {John Postl}. 2015.
	\newblock \showarticletitle{Approximating Optimal Social Choice under Metric
		Preferences.}. In {\em AAAI '15}. 777--783.
	\newblock
	
	
	\bibitem[\protect\citeauthoryear{Ballester and Haeringer}{Ballester and
		Haeringer}{2011}]%
	{ballester2011characterization}
	{Miguel~A Ballester} {and} {Guillaume Haeringer}. 2011.
	\newblock \showarticletitle{A characterization of the single-peaked domain}.
	\newblock {\em Social Choice and Welfare\/} {36}, 2 (2011), 305--322.
	\newblock
	
	
	\bibitem[\protect\citeauthoryear{Barber{\`a}, Gul, and Stacchetti}{Barber{\`a}
		et~al\mbox{.}}{1993}]%
	{barbera1993generalized}
	{Salvador Barber{\`a}}, {Faruk Gul}, {and} {Ennio Stacchetti}. 1993.
	\newblock \showarticletitle{Generalized median voter schemes and committees}.
	\newblock {\em Journal of Economic Theory\/} {61}, 2 (1993), 262--289.
	\newblock
	
	
	\bibitem[\protect\citeauthoryear{Bartholdi and Trick}{Bartholdi and
		Trick}{1986}]%
	{bartholdi1986stable}
	{John Bartholdi} {and} {Michael~A Trick}. 1986.
	\newblock \showarticletitle{Stable matching with preferences derived from a
		psychological model}.
	\newblock {\em Operations Research Letters\/} {5}, 4 (1986), 165--169.
	\newblock
	
	
	\bibitem[\protect\citeauthoryear{Basu, Pollack, and Roy}{Basu
		et~al\mbox{.}}{1996}]%
	{basu1996combinatorial}
	{Saugata Basu}, {Richard Pollack}, {and} {Marie-Fran{\c{c}}oise Roy}. 1996.
	\newblock \showarticletitle{On the combinatorial and algebraic complexity of
		quantifier elimination}.
	\newblock {\it J. ACM} {43}, 6 (1996), 1002--1045.
	\newblock
	
	
	\bibitem[\protect\citeauthoryear{Bennett and Hays}{Bennett and Hays}{1960}]%
	{bennett1960multidimensional}
	{Joseph~F Bennett} {and} {William~L Hays}. 1960.
	\newblock \showarticletitle{Multidimensional unfolding: Determining the
		dimensionality of ranked preference data}.
	\newblock {\em Psychometrika\/} {25}, 1 (1960), 27--43.
	\newblock
	
	
	\bibitem[\protect\citeauthoryear{Betzler, Slinko, and Uhlmann}{Betzler
		et~al\mbox{.}}{2013}]%
	{betzler2013computation}
	{Nadja Betzler}, {Arkadii Slinko}, {and} {Johannes Uhlmann}. 2013.
	\newblock \showarticletitle{On the computation of fully proportional
		representation}.
	\newblock {\em Journal of Artificial Intelligence Research\/} {47}, 1 (2013),
	475--519.
	\newblock
	
	
	\bibitem[\protect\citeauthoryear{Bogomolnaia and Laslier}{Bogomolnaia and
		Laslier}{2007}]%
	{bogomolnaia2007euclidean}
	{Anna Bogomolnaia} {and} {Jean-Fran{\c{c}}ois Laslier}. 2007.
	\newblock \showarticletitle{Euclidean preferences}.
	\newblock {\em Journal of Mathematical Economics\/} {43}, 2 (2007), 87--98.
	\newblock
	
	
	\bibitem[\protect\citeauthoryear{Bokowski and Sturmfels}{Bokowski and
		Sturmfels}{1989}]%
	{bokowski1989infinite}
	{J{\"u}rgen Bokowski} {and} {Bernd Sturmfels}. 1989.
	\newblock \showarticletitle{An infinite family of minor-minimal nonrealizable
		3-chirotopes}.
	\newblock {\em Mathematische Zeitschrift\/} {200}, 4 (1989), 583--589.
	\newblock
	
	
	\bibitem[\protect\citeauthoryear{Borg and Groenen}{Borg and Groenen}{2005}]%
	{borg2005modern}
	{Ingwer Borg} {and} {Patrick~JF Groenen}. 2005.
	\newblock {\em Modern multidimensional scaling: Theory and applications}.
	\newblock Springer Science \& Business Media.
	\newblock
	
	
	\bibitem[\protect\citeauthoryear{Brandt, Brill, Hemaspaandra, and
		Hemaspaandra}{Brandt et~al\mbox{.}}{2010}]%
	{brandt2010bypassing}
	{Felix Brandt}, {Markus Brill}, {Edith Hemaspaandra}, {and} {Lane~A
		Hemaspaandra}. 2010.
	\newblock \showarticletitle{Bypassing Combinatorial Protections:
		Polynomial-Time Algorithms for Single-Peaked Electorates}. In {\em AAAI '10}.
	\newblock
	
	
	\bibitem[\protect\citeauthoryear{Bredereck, Chen, and Woeginger}{Bredereck
		et~al\mbox{.}}{2013a}]%
	{bre-che-woe:c:nearly-restricted}
	{Robert Bredereck}, {Jiehua Chen}, {and} {Gerhard~J Woeginger}. 2013a.
	\newblock \showarticletitle{Are there any nicely structured preference profiles
		nearby?}. In {\em IJCAI '13}. 62--68.
	\newblock
	
	
	\bibitem[\protect\citeauthoryear{Bredereck, Chen, and Woeginger}{Bredereck
		et~al\mbox{.}}{2013b}]%
	{bredereck2013characterization}
	{Robert Bredereck}, {Jiehua Chen}, {and} {Gerhard~J Woeginger}. 2013b.
	\newblock \showarticletitle{A characterization of the single-crossing domain}.
	\newblock {\em Social Choice and Welfare\/} {41}, 4 (2013), 989--998.
	\newblock
	
	
	\bibitem[\protect\citeauthoryear{Bruner and Lackner}{Bruner and
		Lackner}{2015}]%
	{arxiv/BrunerL-likelihoodSP}
	{Marie-Louise Bruner} {and} {Martin Lackner}. 2015.
	\newblock {\em On the Likelihood of Single-Peaked Preferences}.
	\newblock {T}echnical {R}eport arXiv:1505.05852 [cs.GT]. arXiv.org.
	\newblock
	
	
	\bibitem[\protect\citeauthoryear{Busing, Groenen, and Heiser}{Busing
		et~al\mbox{.}}{2005}]%
	{busing2005avoiding}
	{Frank~MTA Busing}, {Patrick~JK Groenen}, {and} {Willem~J Heiser}. 2005.
	\newblock \showarticletitle{Avoiding degeneracy in multidimensional unfolding
		by penalizing on the coefficient of variation}.
	\newblock {\em Psychometrika\/} {70}, 1 (2005), 71--98.
	\newblock
	
	
	\bibitem[\protect\citeauthoryear{Canny}{Canny}{1988}]%
	{canny1988some}
	{John Canny}. 1988.
	\newblock \showarticletitle{Some algebraic and geometric computations in
		PSPACE}. In {\em STOC '88}. ACM, 460--467.
	\newblock
	
	
	\bibitem[\protect\citeauthoryear{Chen, Pruhs, and Woeginger}{Chen
		et~al\mbox{.}}{2015}]%
	{ChenPW15}
	{Jiehua Chen}, {Kirk Pruhs}, {and} {Gerhard~J. Woeginger}. 2015.
	\newblock {\em The one-dimensional Euclidean domain: Finitely many obstructions
		are not enough}.
	\newblock {T}echnical {R}eport arXiv:1506.03838 [cs.GT]. arXiv.org.
	\newblock
	
	
	\bibitem[\protect\citeauthoryear{Coombs}{Coombs}{1950}]%
	{coombs1950psychological}
	{Clyde~H Coombs}. 1950.
	\newblock \showarticletitle{Psychological scaling without a unit of
		measurement.}
	\newblock {\em Psychological review\/} {57}, 3 (1950), 145.
	\newblock
	
	
	\bibitem[\protect\citeauthoryear{Coombs}{Coombs}{1964}]%
	{coombs1964theory}
	{Clyde~H Coombs}. 1964.
	\newblock {\em A Theory of Data}.
	\newblock John Wiley \& Sons.
	\newblock
	\showLCCN{63020629}
	
	
	\bibitem[\protect\citeauthoryear{Cornaz, Galand, and Spanjaard}{Cornaz
		et~al\mbox{.}}{2012}]%
	{cor-gal-spa:c:spwidth}
	{Denis Cornaz}, {Lucie Galand}, {and} {Olivier Spanjaard}. 2012.
	\newblock \showarticletitle{Bounded Single-Peaked Width and Proportional
		Representation}. In {\em ECAI '12}. 270--275.
	\newblock
	
	
	\bibitem[\protect\citeauthoryear{De~Moura and Bj{\o}rner}{De~Moura and
		Bj{\o}rner}{2008}]%
	{de2008z3}
	{Leonardo De~Moura} {and} {Nikolaj Bj{\o}rner}. 2008.
	\newblock \showarticletitle{Z3: An efficient SMT solver}.
	\newblock In {\em Tools and Algorithms for the Construction and Analysis of
		Systems}. Springer, 337--340.
	\newblock
	
	
	\bibitem[\protect\citeauthoryear{Doignon and Falmagne}{Doignon and
		Falmagne}{1994}]%
	{doignon1994polynomial}
	{Jean-Paul Doignon} {and} {Jean-Claude Falmagne}. 1994.
	\newblock \showarticletitle{A polynomial time algorithm for unidimensional
		unfolding representations}.
	\newblock {\em Journal of Algorithms\/} {16}, 2 (1994), 218--233.
	\newblock
	
	
	\bibitem[\protect\citeauthoryear{Eguia}{Eguia}{2011}]%
	{eguia2011foundations}
	{Jon~X Eguia}. 2011.
	\newblock \showarticletitle{Foundations of spatial preferences}.
	\newblock {\em Journal of Mathematical Economics\/} {47}, 2 (2011), 200--205.
	\newblock
	
	
	\bibitem[\protect\citeauthoryear{Elkind and Faliszewski}{Elkind and
		Faliszewski}{2014}]%
	{elkind2014recognizing}
	{Edith Elkind} {and} {Piotr Faliszewski}. 2014.
	\newblock \showarticletitle{Recognizing 1-Euclidean preferences: An alternative
		approach}.
	\newblock In {\em Algorithmic Game Theory}. Springer, 146--157.
	\newblock
	
	
	\bibitem[\protect\citeauthoryear{Elkind, Faliszewski, and Skowron}{Elkind
		et~al\mbox{.}}{2014}]%
	{elkind2014characterization}
	{Edith Elkind}, {Piotr Faliszewski}, {and} {Piotr Skowron}. 2014.
	\newblock \showarticletitle{A Characterization of the Single-Peaked
		Single-Crossing Domain}. In {\em AAAI '14}. 654--660.
	\newblock
	
	
	\bibitem[\protect\citeauthoryear{Elkind, Faliszewski, and Slinko}{Elkind
		et~al\mbox{.}}{2012}]%
	{elkind2012clone}
	{Edith Elkind}, {Piotr Faliszewski}, {and} {Arkadii Slinko}. 2012.
	\newblock \showarticletitle{Clone structures in voters' preferences}. In {\em
		EC '12}. ACM, 496--513.
	\newblock
	
	
	\bibitem[\protect\citeauthoryear{Elkind and Lackner}{Elkind and
		Lackner}{2014}]%
	{elkind2014detecting}
	{Edith Elkind} {and} {Martin Lackner}. 2014.
	\newblock \showarticletitle{On Detecting Nearly Structured Preference
		Profiles}. In {\em AAAI '14}.
	\newblock
	
	
	\bibitem[\protect\citeauthoryear{Elkind and Lackner}{Elkind and
		Lackner}{2015}]%
	{elkind2015structured}
	{Edith Elkind} {and} {Martin Lackner}. 2015.
	\newblock \showarticletitle{Structure in Dichotomous Preferences}. In {\em
		IJCAI '15}. 2019--2025.
	\newblock
	\showISBNx{978-1-57735-738-4}
	
	
	\bibitem[\protect\citeauthoryear{Erd{\'e}lyi, Lackner, and
		Pfandler}{Erd{\'e}lyi et~al\mbox{.}}{2013}]%
	{erd-lac-pfa:c:nearly-sp}
	{G{\'a}bor Erd{\'e}lyi}, {Martin Lackner}, {and} {Andreas Pfandler}. 2013.
	\newblock \showarticletitle{The Complexity of Nearly Single-Peaked
		Consistency}. In {\em AAAI '13}. 283--289.
	\newblock
	
	
	\bibitem[\protect\citeauthoryear{Escoffier, Lang, and {\"O}zt{\"u}rk}{Escoffier
		et~al\mbox{.}}{2008}]%
	{escoffier2008single}
	{Bruno Escoffier}, {J{\'e}r{\^o}me Lang}, {and} {Meltem {\"O}zt{\"u}rk}. 2008.
	\newblock \showarticletitle{Single-peaked consistency and its complexity.}. In
	{\em ECAI '08}, Vol.~8. 366--370.
	\newblock
	
	
	\bibitem[\protect\citeauthoryear{Faliszewski, Hemaspaandra, Hemaspaandra, and
		Rothe}{Faliszewski et~al\mbox{.}}{2009}]%
	{faliszewski2009shield}
	{Piotr Faliszewski}, {Edith Hemaspaandra}, {Lane~A Hemaspaandra}, {and}
	{J{\"o}rg Rothe}. 2009.
	\newblock \showarticletitle{The shield that never was: Societies with
		single-peaked preferences are more open to manipulation and control}. In {\em
		TARK '09}. ACM, 118--127.
	\newblock
	
	
	\bibitem[\protect\citeauthoryear{Garg, Mehta, Vazirani, and Yazdanbod}{Garg
		et~al\mbox{.}}{2015}]%
	{garg2015etr}
	{Jugal Garg}, {Ruta Mehta}, {Vijay~V Vazirani}, {and} {Sadra Yazdanbod}. 2015.
	\newblock \showarticletitle{ETR-Completeness for Decision Versions of
		Multi-player (Symmetric) Nash Equilibria}.
	\newblock In {\em Automata, Languages, and Programming}. Springer, 554--566.
	\newblock
	
	
	\bibitem[\protect\citeauthoryear{Grandmont}{Grandmont}{1978}]%
	{grandmont1978intermediate}
	{Jean-Michel Grandmont}. 1978.
	\newblock \showarticletitle{Intermediate Preferences and the Majority Rule}.
	\newblock {\em Econometrica\/} {46}, 2 (1978), 317--30.
	\newblock
	
	
	\bibitem[\protect\citeauthoryear{Grigor'ev}{Grigor'ev}{1988}]%
	{grigor1988complexity}
	{D~Yu Grigor'ev}. 1988.
	\newblock \showarticletitle{Complexity of deciding Tarski algebra}.
	\newblock {\em Journal of Symbolic Computation\/} {5}, 1 (1988), 65--108.
	\newblock
	
	
	\bibitem[\protect\citeauthoryear{Hays and Bennett}{Hays and Bennett}{1961}]%
	{hays1961multidimensional}
	{William~L Hays} {and} {Joseph~F Bennett}. 1961.
	\newblock \showarticletitle{Multidimensional unfolding: Determining
		configuration from complete rank order preference data}.
	\newblock {\em Psychometrika\/} {26}, 2 (1961), 221--238.
	\newblock
	
	
	\bibitem[\protect\citeauthoryear{Hotelling}{Hotelling}{1929}]%
	{hotelling1929stability}
	{Harold Hotelling}. 1929.
	\newblock \showarticletitle{Stability in competition}.
	\newblock {\em The Economic Journal\/} {39}, 153 (1929), 41--57.
	\newblock
	
	
	\bibitem[\protect\citeauthoryear{Jovanovi{\'c} and De~Moura}{Jovanovi{\'c} and
		De~Moura}{2012}]%
	{jovanovic2012solving}
	{Dejan Jovanovi{\'c}} {and} {Leonardo De~Moura}. 2012.
	\newblock \showarticletitle{Solving non-linear arithmetic}.
	\newblock In {\em Automated Reasoning}. Springer, 339--354.
	\newblock
	
	
	\bibitem[\protect\citeauthoryear{Kamiya, Orlik, Takemura, and Terao}{Kamiya
		et~al\mbox{.}}{2006}]%
	{kamiya2006arrangements}
	{Hidehiko Kamiya}, {Peter Orlik}, {Akimichi Takemura}, {and} {Hiroaki Terao}.
	2006.
	\newblock \showarticletitle{Arrangements and ranking patterns}.
	\newblock {\em Annals of Combinatorics\/} {10}, 2 (2006), 219--235.
	\newblock
	
	
	\bibitem[\protect\citeauthoryear{Kamiya, Takemura, and Terao}{Kamiya
		et~al\mbox{.}}{2011}]%
	{kamiya2011codimension}
	{Hidehiko Kamiya}, {Akimichi Takemura}, {and} {Hiroaki Terao}. 2011.
	\newblock \showarticletitle{Ranking patterns of unfolding models of codimension
		one}.
	\newblock {\em Advances in Applied Mathematics\/} {47}, 2 (2011), 379--400.
	\newblock
	
	
	\bibitem[\protect\citeauthoryear{Kang and M{\"u}ller}{Kang and
		M{\"u}ller}{2012}]%
	{kang2012sphere}
	{Ross~J Kang} {and} {Tobias M{\"u}ller}. 2012.
	\newblock \showarticletitle{Sphere and dot product representations of graphs}.
	\newblock {\em Discrete \& Computational Geometry\/} {47}, 3 (2012), 548--568.
	\newblock
	
	
	\bibitem[\protect\citeauthoryear{Knoblauch}{Knoblauch}{2010}]%
	{knoblauch2010recognizing}
	{Vicki Knoblauch}. 2010.
	\newblock \showarticletitle{Recognizing one-dimensional Euclidean preference
		profiles}.
	\newblock {\em Journal of Mathematical Economics\/} {46}, 1 (2010), 1--5.
	\newblock
	
	
	\bibitem[\protect\citeauthoryear{Kruskal and Carroll}{Kruskal and
		Carroll}{1969}]%
	{kruskal1969geometrical}
	{Joseph~B Kruskal} {and} {J~Douglas Carroll}. 1969.
	\newblock \showarticletitle{Geometrical models and badness-of-fit functions}.
	\newblock {\em Multivariate Analysis\/}  {2} (1969), 639--671.
	\newblock
	
	
	\bibitem[\protect\citeauthoryear{Lekkerkerker and Boland}{Lekkerkerker and
		Boland}{1962}]%
	{lekkeikerker1962representation}
	{Cornelis~G Lekkerkerker} {and} {J Boland}. 1962.
	\newblock \showarticletitle{Representation of a finite graph by a set of
		intervals on the real line}.
	\newblock {\em Fundamenta Mathematicae\/} {51}, 1 (1962), 45--64.
	\newblock
	
	
	\bibitem[\protect\citeauthoryear{Mattei and Walsh}{Mattei and Walsh}{2013}]%
	{mattei2013preflib}
	{Nicholas Mattei} {and} {Toby Walsh}. 2013.
	\newblock \showarticletitle{Preflib: A library for preferences http://www.
		preflib. org}.
	\newblock In {\em Algorithmic Decision Theory}. Springer, 259--270.
	\newblock
	
	
	\bibitem[\protect\citeauthoryear{McDiarmid and M{\"u}ller}{McDiarmid and
		M{\"u}ller}{2013}]%
	{mcdiarmid2013integer}
	{Colin McDiarmid} {and} {Tobias M{\"u}ller}. 2013.
	\newblock \showarticletitle{Integer realizations of disk and segment graphs}.
	\newblock {\em Journal of Combinatorial Theory, Series B\/} {103}, 1 (2013),
	114--143.
	\newblock
	
	
	\bibitem[\protect\citeauthoryear{Merrill and Grofman}{Merrill and
		Grofman}{1999}]%
	{merrill1999unified}
	{Samuel Merrill} {and} {Bernard Grofman}. 1999.
	\newblock {\em A unified theory of voting: Directional and proximity spatial
		models}.
	\newblock Cambridge University Press.
	\newblock
	
	
	\bibitem[\protect\citeauthoryear{Mn{\"e}v}{Mn{\"e}v}{1985}]%
	{mnev1985realizability}
	{Nikolai~E Mn{\"e}v}. 1985.
	\newblock \showarticletitle{Realizability of combinatorial types of convex
		polyhedra over fields}.
	\newblock {\em Journal of Soviet Mathematics\/} {28}, 4 (1985), 606--609.
	\newblock
	
	
	\bibitem[\protect\citeauthoryear{Peters and Elkind}{Peters and Elkind}{2016}]%
	{peters2016preferences}
	{Dominik Peters} {and} {Edith Elkind}. 2016.
	\newblock \showarticletitle{Preferences Single-Peaked on Nice Trees}. In {\em
		AAAI '16}.
	\newblock
	
	
	\bibitem[\protect\citeauthoryear{Renegar}{Renegar}{1992}]%
	{renegar1992computational}
	{James Renegar}. 1992.
	\newblock \showarticletitle{On the computational complexity and geometry of the
		first-order theory of the reals. Part I: Introduction. Preliminaries. The
		geometry of semi-algebraic sets. The decision problem for the existential
		theory of the reals}.
	\newblock {\em Journal of Symbolic Computation\/} {13}, 3 (1992), 255--299.
	\newblock
	
	
	\bibitem[\protect\citeauthoryear{Roskam}{Roskam}{1968}]%
	{roskam1968metric}
	{Edwarda Elias Charles~Iben Roskam}. 1968.
	\newblock {\em Metric analysis of ordinal data in psychology}.
	\newblock Ph.D. Dissertation.
	\newblock
	
	
	\bibitem[\protect\citeauthoryear{Schaefer}{Schaefer}{2010}]%
	{schaefer2010complexity}
	{Marcus Schaefer}. 2010.
	\newblock \showarticletitle{Complexity of Some Geometric and Topological
		Problems}. In {\em Graph Drawing: 17th International Symposium}, Vol. 5849.
	Springer, 334--344.
	\newblock
	
	
	\bibitem[\protect\citeauthoryear{Schaefer}{Schaefer}{2013}]%
	{schaefer2013realizability}
	{Marcus Schaefer}. 2013.
	\newblock \showarticletitle{Realizability of graphs and linkages}.
	\newblock In {\em Thirty Essays on Geometric Graph Theory}. Springer, 461--482.
	\newblock
	
	
	\bibitem[\protect\citeauthoryear{Schaefer and
		{\smash{\v{S}}}tefankovi{\v{c}}}{Schaefer and
		{\smash{\v{S}}}tefankovi{\v{c}}}{2015}]%
	{SchaeferNash2015}
	{Marcus Schaefer} {and} {Daniel {\smash{\v{S}}}tefankovi{\v{c}}}. 2015.
	\newblock \showarticletitle{Fixed Points, Nash Equilibria, and the Existential
		Theory of the Reals}.
	\newblock {\em Theory of Computing Systems\/} (2015), 1--22.
	\newblock
	\showISSN{1433-0490}
	
	
	\bibitem[\protect\citeauthoryear{Shor}{Shor}{1991}]%
	{shor1991stretchability}
	{Peter Shor}. 1991.
	\newblock \showarticletitle{Stretchability of pseudolines is NP-hard}.
	\newblock {\em Applied Geometry and Discrete Mathematics, Amer. Math. Soc.,
		Providence, RI\/}  {4} (1991), 531--554.
	\newblock
	
	
	\bibitem[\protect\citeauthoryear{Sui, Francois-Nienaber, and Boutilier}{Sui
		et~al\mbox{.}}{2013}]%
	{sui2013multi}
	{Xin Sui}, {Alex Francois-Nienaber}, {and} {Craig Boutilier}. 2013.
	\newblock \showarticletitle{Multi-dimensional single-peaked consistency and its
		approximations}. In {\em IJCAI '13}. 375--382.
	\newblock
	
	
	\bibitem[\protect\citeauthoryear{Takane, Young, and De~Leeuw}{Takane
		et~al\mbox{.}}{1977}]%
	{takane1977nonmetric}
	{Yoshio Takane}, {Forrest~W Young}, {and} {Jan De~Leeuw}. 1977.
	\newblock \showarticletitle{Nonmetric individual differences multidimensional
		scaling: An alternating least squares method with optimal scaling features}.
	\newblock {\em Psychometrika\/} {42}, 1 (1977), 7--67.
	\newblock
	
	
	\bibitem[\protect\citeauthoryear{Tardos}{Tardos}{1986}]%
	{tardos1986strongly}
	{Eva Tardos}. 1986.
	\newblock \showarticletitle{A strongly polynomial algorithm to solve
		combinatorial linear programs}.
	\newblock {\em Operations Research\/} {34}, 2 (1986), 250--256.
	\newblock
	
	
	\bibitem[\protect\citeauthoryear{Tarski}{Tarski}{1948}]%
	{tarski1948decision}
	{Alfred Tarski}. 1948.
	\newblock {\em A decision method for elementary algebra and geometry}.
	\newblock Rand Corporation.
	\newblock
	
	
	\bibitem[\protect\citeauthoryear{Trick}{Trick}{1989}]%
	{trick1989recognizing}
	{Michael~A Trick}. 1989.
	\newblock \showarticletitle{Recognizing single-peaked preferences on a tree}.
	\newblock {\em Mathematical Social Sciences\/} {17}, 3 (1989), 329--334.
	\newblock
	
	
	\bibitem[\protect\citeauthoryear{Tucker}{Tucker}{1972}]%
	{tucker1972structure}
	{Alan Tucker}. 1972.
	\newblock \showarticletitle{A structure theorem for the consecutive 1's
		property}.
	\newblock {\em Journal of Combinatorial Theory, Series B\/} {12}, 2 (1972),
	153--162.
	\newblock
	
	
\end{thebibliography}

\end{document}